\documentclass[11pt,a4paper,reqno]{amsart}%
\usepackage{amsthm,amsmath,amsfonts,amssymb,amsxtra,appendix,bookmark,euscript,delarray,dsfont}
\usepackage[latin1]{inputenc}

\makeatletter
\def\@tocline#1#2#3#4#5#6#7{\relax
  \ifnum #1>\c@tocdepth 
  \else
    \par \addpenalty\@secpenalty\addvspace{#2}%
    \begingroup \hyphenpenalty\@M
    \@ifempty{#4}{%
      \@tempdima\csname r@tocindent\number#1\endcsname\relax
    }{%
      \@tempdima#4\relax
    }%
    \parindent\z@ \leftskip#3\relax \advance\leftskip\@tempdima\relax
    \rightskip\@pnumwidth plus4em \parfillskip-\@pnumwidth
    #5\leavevmode\hskip-\@tempdima
      \ifcase #1
       \or\or \hskip 1em \or \hskip 2em \else \hskip 3em \fi%
      #6\nobreak\relax
      \dotfill
      \hbox to\@pnumwidth{\@tocpagenum{#7}}
    \par
    \nobreak
    \endgroup
  \fi}
\makeatother

\theoremstyle{plain}
\newtheorem{theorem}{Theorem}[section]
\newtheorem{lemma}[theorem]{Lemma}

\theoremstyle{remark}
\newtheorem{remark}[theorem]{Remark}

\numberwithin{equation}{section}

\newcommand\R{{\ensuremath {\mathbb R} }}
\newcommand\C{{\ensuremath {\mathbb C} }}
\newcommand\N{{\ensuremath {\mathbb N} }}

\renewcommand\phi{\varphi}

\newcommand{\wto}{\rightharpoonup}

\newcommand{\cE}{\mathcal{E}}

\newcommand{\eps}{\epsilon}

\renewcommand{\epsilon}{\varepsilon}

\renewcommand{\geq}{\geqslant}
\renewcommand{\leq}{\leqslant}
\renewcommand{\hat}{\widehat}
\renewcommand{\tilde}{\widetilde}

\newcommand{\Udi}{U_{\mathrm{dip}}}

\title[Stability of 2D dipolar BECs]{On the stability of 2D dipolar Bose-Einstein condensates}

\author[A. Eychenne]{Arnaud EYCHENNE}
\address{Universit\'e Grenoble-Alpes \& CNRS,  LPMMC (UMR 5493), B.P. 166, F-38042 Grenoble, France}
\email{arnaud.eychenne@u-psud.fr}

\author[N. Rougerie]{Nicolas ROUGERIE}
\address{Universit\'e Grenoble-Alpes \& CNRS,  LPMMC (UMR 5493), B.P. 166, F-38042 Grenoble, France}
\email{nicolas.rougerie@grenoble.cnrs.fr}

\begin{document}
\date{September, 2018}

\begin{abstract}
We study the existence of energy minimizers for a Bose-Einstein condensate with dipole-dipole interactions, tightly confined to a plane. The problem is critical in that the kinetic energy and the (partially attractive) interaction energy behave the same under mass-preserving scalings of the wave-function. We obtain a sharp criterion for the existence of ground states, involving the optimal constant of a certain generalized Gagliardo-Nirenberg inequality.     
\end{abstract}

\maketitle

\setcounter{tocdepth}{2}
\tableofcontents

\section{Introduction}

Bose-Einstein condensates (BEC) are macroscopic matter waves consisting of a large number of particles occupying the same quantum state~\cite{CorWie-nobel,Ketterle-nobel}. Since the first realization of BECs with large magnetic dipole moment~\cite{GriPfa-05,StuPfa-05}, dipolar BECs have become a popular subject of investigation in the cold atoms physics community, see~\cite{LahPfa-09} for review.    

In typical experiments, all the dipoles are polarized along a fixed, common, direction set by a strong external field. A 3D condensate with dipole-dipole interactions is then modeled via the following non-local and non-linear Schr\"odinger equation: 
\begin{equation}\label{eq:3D NLS}
i\partial_t \psi = -\frac{1}{2} \Delta \psi + V \psi + \beta |\psi| ^2 \psi + \lambda (\Udi \star |\psi| ^2) \psi 
\end{equation}
where $\psi:\R^3 \mapsto \C$ is the common wave-function of the Bose-condensed particles, $V:\R^3 \mapsto \R^+$ is a confining potential, $\beta,\lambda \in \R$ are the strength of van der Waals (modeled as zero-range, contact) interactions and dipole-dipole interactions, respectively. The dipole-dipole interaction potential is  
\begin{equation}\label{eq:3D dip dip}
\Udi (x) = \frac{3}{4\pi} \frac{1 - 3 (x \cdot n) ^2 / |x| ^2}{|x|^3} = \frac{3}{4\pi} \frac{1 - 3 \cos ^2 (\theta_x) }{|x|^3} 
\end{equation}
where the unit vector $n$ denotes the fixed polarization direction and $\theta_x$ the angle between $n$ and $x$. Note the long-range and anisotropic nature of this potential, which contrasts with the more standard contact interactions. 

The above equation conserves the mass, that we shall henceforth fix to unity:
\begin{equation}\label{eq:3D mass}
\int_{\R^3} |\psi| ^2 = 1,  
\end{equation}
and the energy 
\begin{equation}\label{eq:3D energy}
\cE_{\rm 3D} [\psi] := \int_{\R^3} \frac{1}{2} |\nabla \psi| ^2 + V |\psi| ^2 + \frac{\beta}{2} |\psi| ^4 + \frac{\lambda}{2} \left( \Udi \star |\psi| ^2 \right)|\psi|^2. 
\end{equation}
A central question is to investigate whether or not the energy is bounded below and has minimizers (ground states) under the mass constraint~\eqref{eq:3D mass}. This we refer to as the stability question. 

It has been proved in~\cite{BaoCaiWan-10,CarHaj-15} (see also~\cite{AntSpa-11,CarMarSpa-08,BelJea-16,BelFor-18,LuoSty-18} for other mathematical studies of 3D dipolar BECs, in particular for the existence of non energy minimizing stationary states) that stability holds if and only if 
$$\beta \geq 0 \mbox{ and } -\frac{\beta}{2} \leq \lambda \leq \beta.$$
In particular, a positive contact interaction is always needed to stabilize the gas, and in case of stability, the total interaction energy must always be  non-negative 
$$
\beta \int_{\R^3}  |\psi| ^4 + \lambda \int_{\R^3} \left( \Udi \star |\psi| ^2 \right)|\psi|^2  \geq 0 \mbox{ for all } \psi
$$
so that the total interaction is always globally repulsive. In the stable regime, the 3D dipolar Gross-Pitaevskii energy has recently been rigorously derived from many-body quantum mechanics in~\cite{Triay-17}.  

Quasi-2D dipolar BECs should be more stable than 3D ones, as  put forward first in~\cite{Fisher-06}. Dimensional reduction indeed makes more room for studying attractive interactions, which is particularly interesting in view of the experimental possibilities to tune the trapping potential $V$ for tight confinement along one or two directions. 1D condensates, with or without dipole-dipole interactions, are always stable, because the kinetic energy then prevents the collapse that might be triggered by attractive interactions. The most interesting case is the borderline 2D one, where kinetic and interaction energies have the same order of magnitude in case of collapse. Our aim in this paper is to investigate the stability question for 2D dipolar BECs, with particular emphasis on the attractive regime. We shall extend the results of~\cite{BaoAbdCai-12} by finding a sharp criterion for stability. 

The model for 2D dipolar BECs we use is that formally derived from 3D in~\cite{CaiRosLeiBao-10,Fisher-06,ParOdel-08} (see also~\cite{PedSan-05,NatPedSan-09}, and~\cite{BaoTreMeh-14b} for mathematically rigorous results). We find that the stability criterion  is given in terms of a certain large-frequency limit of the 2D dipolar interaction. In this limit, the  interaction energy reduces to one of the form 
\begin{equation}\label{eq:2D int high}
-a \int_{\R^2} |\psi| ^4 -b \int_{\R^2} (U ^{\rm 2D} \star |\psi| ^2)  |\psi| ^2   
\end{equation}
where $a$ and $b$ are suitable effective parameters and $U ^{\rm 2D}$ is a close analogue of the 3D dipole-dipole interaction~\eqref{eq:3D dip dip}:
\begin{equation}\label{eq:2D dip high}
U ^{\rm 2D} (x) = \frac{1 - 2 \cos ^2 (\theta_x)}{2\pi|x|^2}
\end{equation}
with $\theta_x$ the angle in polar coordinates. Let then $C(a,b)$ be the optimal constant in the following variant of the Gagliardo-Nirenberg inequality\footnote{It might be that $C(a,b)=0$, see below.}:
\begin{equation}\label{eq:intro GN}
a \int_{\R^2} |\psi| ^4 + b \int_{\R^2} (U ^{\rm 2D} \star |\psi| ^2)  |\psi| ^2 \leq C \int_{\R^2} |\nabla \psi| ^2 \int_{\R^2} |\psi| ^2.
\end{equation}
The sharp stability criterion we find is 
$$ C(a,b) < 1, $$
which should be compared with the well-known stability criterion for 2D BECs without dipole-dipole interactions~\cite{GuoSei-13,Maeda-10,Weinstein-83,Zhang-00}, involving the usual ($b = 0$) Gagliardo-Nirenberg inequality and its optimizers~\cite{Frank-13}. The above stability criterion is not very explicit, but in the particular case of dipoles polarized perpendicular to the plane of confinement, it turns out that $b = 0$, so that a simpler expression is obtained. We also investigate the borderline situation $C(a,b) = 1$. It differs markedly from the purely non-dipolar situation thoroughly investigated in~\cite{GuoSei-13} and following papers~\cite{GuoLinWei-17,GuoZenZho-16,LewNamRou-18a} (see also~\cite{AshFroGraSchTro-02,LieYau-87,FroJonLen-07b,Nguyen-18} for related topics), for then the next order of the original, physical, 2D interaction matters in the high-frequency limit. 

\medskip

\subsubsection*{\bf Acknowledgement} 
This project has received funding from the European Research Council (ERC) under the European Union's Horizon 2020 research and innovation programme (grant agreement CORFRONMAT No 758620).

\section{Main Results}

To obtain a sensible 2D model of a dipolar BEC, the procedure followed in~\cite{CaiRosLeiBao-10,ParOdel-08} is to start from the 3D model~\eqref{eq:3D NLS}, \eqref{eq:3D energy} and insert there a potential $V$ strongly confining along the $x_3$ direction. The most physically relevant case is to consider a harmonic perpendicular confinement:
$$ V(x_1,x_2,x_3) = \tilde{V} (x_1,x_2) + \frac{x_3 ^2}{2\eps ^4} $$
with $\eps >0$ a small parameter. It is then natural to use an ansatz of the form 
\begin{equation}\label{eq:2D ansatz}
 \psi (x_1,x_2,x_3) = u (x_1,x_2) \frac{1}{\eps^{1/2} \pi ^{1/4}}e^{-\frac{x_3 ^2}{2\eps ^2}} 
\end{equation}
for the wave-function of the 3D condensate, i.e. to assume that the motion in the $x_3$-direction is confined to the ground state of the harmonic oscillator. The question is then: what is the effective 2D interaction felt by $u$ ? The computation is facilitated by noting that the 3D interaction satisfies
$$ \Udi = - \delta_0 - \frac{3}{4\pi} \partial_{nn} ^2 \left( \frac{1}{|\,.\,|} \right)$$
in the sense of distributions. The conclusion is that inserting the ansatz~\eqref{eq:2D ansatz} in the energy~\eqref{eq:3D energy} leads to the effective functional 
\begin{equation}\label{eq:2D energy I}
\cE_{\rm 2D} [u] = \int_{\R^2} \frac{1}{2} |\nabla u| ^2 + \tilde{V} |u|^2 + \frac{\beta - \lambda + 3 n_3 ^2 \lambda}{2\sqrt{2\pi} \eps} |u| ^4 -  \frac{3\lambda}{4} |u| ^2 \Phi 
\end{equation}
where 
\begin{align}\label{eq:2D int I}  
\Phi &= \left( \partial_{n_\perp} ^2 - n_3 ^2 \Delta \right) \left( U^{\rm 2D}_{\eps} \star |u| ^2 \right)\nonumber \\
U^{\rm 2D}_{\eps} (x_1,x_2) &= \frac{1}{2\sqrt{2} \pi^{3/2}} \int_{\R} \frac{e^{-s^2/2}}{\sqrt{x_1 ^2 + x_2 ^2 + \eps ^2 s ^2}} ds.
\end{align}
In the above we have denoted $n = (n_\perp,n_3)$ with $n_\perp \in \R^2, n_3 \in \R$. Recall that the polarization direction $n$ is a unit vector so that $|n_\perp| ^2 + n_3 ^2 = 1$. Note that the $\eps \to 0$ limit is not taken in the above, it is just assumed that $\eps$ is small enough for the ansatz~\eqref{eq:2D ansatz} to make sense. The $\eps \to 0$ limit actually leads to another model whose stability properties are worked out in~\cite[Section~3]{BaoAbdCai-12}. This case is less interesting than the finite $\eps$ one, for the $\eps \to 0$ limit leads to a stability criterion resembling the 3D one, and in particular demands a globally repulsive interaction.   

From now on we focus on the minimization of the above 2D energy. We clean the notation a little bit by setting $\eps = (2\pi) ^{-1/2}$ and rotating the $(x_1,x_2)$ frame in order to have $n_{\perp} = (1-n_3^2, 0)$. This only changes the trapping potential (unless it is radial) and leads to the general problem of minimizing
\begin{equation}\label{eq:2D energy II}
\cE_{\rm 2D} [u] = \int_{\R^2} \frac{1}{2} |\nabla u| ^2 + V |u|^2 + \frac{\beta - \lambda + 3 n_3 ^2 \lambda}{2} |u| ^4 -  \frac{3\lambda}{4} |u| ^2 \Phi 
\end{equation}
under a unit mass constraint,
with
\begin{align}\label{eq:2D int II}  
\Phi &= \left( (1- 2 n_3) ^3 \partial_{x_1} ^2 - n_3 ^2 \partial_{x_2} ^2 \right) \left( K \star |u| ^2 \right)\nonumber \\
K (x_1,x_2) &= \frac{1}{2\sqrt{2} \pi^{3/2}} \int_{\R} \frac{e^{-s^2/2}}{\sqrt{x_1 ^2 + x_2 ^2 + 2\pi s ^2}} ds
\end{align}
and $V$ a (say smooth, for simplicity) trapping potential
$$ V(x) \underset{|x| \to \infty}{\to} \infty.$$
Concerning the minimization problem 
\begin{equation}\label{eq:2D min ener}
E_{\rm 2D} := \inf \left\{ \cE_{\rm 2D} [u] \: \big| \: \int_{\R^2} |u|^2 = 1 \right\},  
\end{equation}
our main result is as follows:

\begin{theorem}[\textbf{Stability/instability of the 2D dipolar gas}]\label{thm:main}\mbox{}\\
\underline{\emph{(1, generalized Gagliardo-Nirenberg inequality)}} Let $a,b$ be a pair of real numbers such that 
\begin{equation}\label{eq:parameters cond}
a + \frac{b}{2} > 0 \mbox{ or } a - \frac{b}{2} > 0
\end{equation}
and $U^{\rm 2D}$ be defined as in~\eqref{eq:2D dip high}. There exists a constant $C(a,b) >0$ such that 
\begin{equation}\label{eq:GN gen}
a \int_{\R ^2} |u| ^4 + b \int_{\R^2} \left( U ^{\rm 2D} \star |u| ^2 \right) |u| ^2 \leq C(a,b) \int_{\R^2} |\nabla u| ^2 \int_{\R^2} |u| ^2    
\end{equation}
for all $u:\R^2 \mapsto \C$, and a function $u_{a,b}$ such that 
\begin{equation}\label{eq:GN gen opt}
a \int_{\R ^2} |u_{a,b}| ^4 + b \int_{\R^2} \left( U ^{\rm 2D} \star |u_{a,b}| ^2 \right) |u_{a,b}| ^2 = C(a,b) \int_{\R^2} |\nabla u_{a,b}| ^2 \int_{\R^2} |u_{a,b}| ^2.
\end{equation}

\medskip 
 
\noindent  \underline{\emph{(2, stability of 2D dipolar BECs)}} In terms of the parameters of~\eqref{eq:2D energy II}, set 
\begin{equation}\label{eq:parameters}
a = \lambda - \beta + \frac{3\lambda}{2} (n_3^2 - 1), \quad b = 3\lambda (n_3^2 - 1) .
\end{equation}
Then, if
\begin{equation}\label{eq:param triv}
a + \frac{b}{2} \leq 0 \mbox{ and } a - \frac{b}{2} \leq 0 
\end{equation}
or if~\eqref{eq:parameters cond} holds and the above constant $C(a,b) < 1$ we have stability, $E_{\rm 2D} > -\infty$ and there exists a minimizer for Problem~\eqref{eq:2D min ener}.   

\medskip

\noindent  \underline{\emph{(3, instability of 2D dipolar BECs)}} Again with the choice~\eqref{eq:parameters}, if~\eqref{eq:parameters cond} holds and $C(a,b) > 1$ we have instability, $E_{\rm 2D} = -\infty.$
\end{theorem}

A few comments: 

\medskip

\noindent\textbf{1.} The above shows that, unlike in 3D, BECs with genuinely attractive dipolar interactions can be stable in 2D. Note that the parameter $b$ defined in~\eqref{eq:parameters} is always nonpositive, but Inequality~\eqref{eq:GN gen} also makes sense for $b\geq 0$. 

\medskip

\noindent\textbf{2.} Note that the convolution appearing in the left side of~\eqref{eq:GN gen} (just as that appearing in~\eqref{eq:3D energy}) is a singular integral, for $U^{\rm 2D}$ is not locally integrable. However, just as $\Udi$, $U^{\rm 2D}$ averages to $0$ on spheres centered at the origin, so that the convolution is well defined as a principal value~\cite[Chapter~4]{Duoandikoetxea-01}. 

\medskip

\noindent\textbf{3.} The non-local interaction term in~\eqref{eq:2D energy II} can be put in a more familiar form
\begin{equation}\label{eq:F int}
 F^{\rm int} [|u| ^2] := \int_{\R^2} |u| ^2 \Phi = \int_{\R^2} \left(\tilde{K}\star |u|^2\right) |u| ^2 
\end{equation}
with an explicit interaction kernel $\tilde{K}$, but the expression, involving Bessel functions, is not particularly illuminating, see~\cite[Appendix~C]{CaiRosLeiBao-10}. Much more useful to us will be the expression in Fourier variables following from~\cite[Lemma~2.2]{BaoAbdCai-12} and Plancherel's formula:
\begin{equation}\label{eq:2D int freq}
 F^{\rm int} [|u|^2] = - \frac{1}{\pi} \int_{\R^2} \int_{\R} \frac{(1-2n_3 ^2)\xi_1 ^2 - n_3 ^2 \xi_2 ^2}{|\xi| ^2 + s^2} |\hat{\rho} (\xi)| ^2 e^{-\frac{s^2}{4\pi} } ds d\xi 
\end{equation}
with the Fourier transform\footnote{Remark that our convention differs from that of~\cite{BaoAbdCai-12}. We prefer the Fourier transform to be an isometry.} 
$$ \hat{f} (\xi) = \frac{1}{2\pi}\int_{\R^2} f(\eta) e^{-i\xi\cdot\eta} d\eta$$
and $\rho = |u| ^2.$ 
 
\medskip

\noindent\textbf{4.} The main insight in the proof is to realize that, should instability occur, it can only be due to a minimizing sequence concentrating around a point. But then, only the small-length/high-frequency part of~\eqref{eq:2D int freq} contributes to the leading order of the energy, which leads to an effective non-local interaction energy (start from~\eqref{eq:2D int freq}, ignore the $s^2$ term in the denominator of the integrand and do the sum in $s$)
$$ 
 \int_{\R^2} |u| ^2 \Phi = - 2 \int_{\R^2} \frac{(1-2n_3 ^2)\xi_1 ^2 - n_3 ^2 \xi_2 ^2}{|\xi| ^2}|\hat{\rho} (\xi)| ^2 d\xi 
$$
that we relate to~\eqref{eq:GN gen} in Lemma~\ref{lem:int high} below.  

\medskip

\noindent\textbf{5.} The above theorem completes the partial results obtained in~\cite[Section~2]{BaoAbdCai-12} by giving an optimal stability criterion. One might prefer more explicit but less optimal criteria. In this direction, note that for dipoles polarized along the perpendicular axis, i.e. $n_3 ^2 = 1$, the generalized Gagliardo-Nirenberg boils down to the usual one (set $a=1,b=0$ in~\eqref{eq:GN gen}):
\begin{equation}\label{eq:GN}
\int_{\R^2} |u| ^4 \leq C \int_{\R^2} |\nabla u| ^2 \int_{\R^2} |u| ^2. 
\end{equation}
Let $C_{\rm GN}$ be the optimal constant~\cite{Weinstein-83,Frank-13} for the above. The stability criterion is that minimizers exist if $\beta - \lambda > C_{\rm GN} ^{-1}$, and the energy is $-\infty$ if $\beta - \lambda < C_{\rm GN} ^{-1}$.  

Note that since $0 \leq n_3^2 \leq 1$, elementary inequalities already used in~\cite{BaoAbdCai-12} allow to deduce more explicit stability/instability results for the general case from the perpendicular dipoles case. These are not optimal however, and for brevity we leave them to the reader.

\bigskip

Theorem~\ref{thm:main} leaves out the equality case where $C(a,b) = 1$. Refining our methods we can however discuss this as well:

\begin{theorem}[\textbf{Borderline cases}]\label{thm:border}\mbox{}\\
With the notation of Theorem~\ref{thm:main}, assume that~\eqref{eq:parameters cond} holds and $C(a,b) = 1$. We have:  

\noindent (1) If $\lambda (1-3 n_3^2) < 0$, then $E ^{2D} > -\infty$ and minimizers exists. 

\medskip

\noindent (2) If $\lambda (1-3 n_3 ^2) = 0$, then $E^{2D} > -\infty$ but there exists sequences $(u_n)_{n\in\N}$ of finite energy that collapse to a point, $|u_n| ^2 \wto \delta_x$ as measures with $\delta_x$ the Dirac mass at some point $x\in \R^2$.  

\medskip 

\noindent (3) If $\lambda (1-3 n_3 ^2) > 0$, then $E^{2D} = -\infty$.

\end{theorem}

Some comments:

\medskip 

\noindent\textbf{1.} This is based on analyzing the next order of the interaction (beyond~\eqref{eq:GN gen}) in the high frequency limit. This leads to a singular term that can either prevent collapse if it is overall repulsive (case (1)) or on the contrary enhance it (case (3)). For the special angle $n_3 ^2 = 1/3$ (case (2)) the second order interaction turns out not to diverge in the high-frequency limit because of an averaging over angles that kills the singularity in Fourier space. We have less information in this case. Note that some sequences having finite energy but badly failing to be compact does not rule out the existence of minimizers by itself. It does rule out standard approaches to the problem however.

\medskip 

\noindent\textbf{2.} In~\cite{GuoSei-13} (case $\lambda = 0$), the behavior of the minimizers is studied in details when one approaches the borderline case from the stability side. It is not obvious how to do the same here because of the diverging next order term in the high frequency limit. 

\bigskip

The rest of the paper is devoted to the proofs of our main theorems. We first discuss the generalized Gagliardo-Nirenberg inequality in Section~\ref{sec:high frequency}, proving Item 1 of Theorem~\ref{thm:main}. Then we deduce Items 2 and 3 in Section~\ref{sec:proof main}. A refinement of the method, discussed in Section~\ref{sec:borderline}, leads to the proof of Theorem~\ref{thm:border}.

\section{The high frequency model}\label{sec:high frequency}

We start by properly defining the high-frequency interaction entering the left-hand side of~\eqref{eq:GN gen}:

\begin{lemma}[\textbf{Interaction energy of collapsing densities}]\label{lem:int high}\mbox{}\\
Recall that 
\begin{equation}\label{eq:2D dip lem}
 U ^{\rm 2D} (x) = \frac{1 - 2 \cos ^2 (\theta_x)}{2\pi|x|^2}. 
\end{equation}
The linear map 
$$ \rho \mapsto U^{\rm 2D} \star \rho$$
is bounded from $L^p (\R^2)$ to $L^p(\R^2)$ for any $1<p<\infty$. Moreover we have the distributional identity 
\begin{equation}\label{eq:distri high}
 U ^{\rm 2D} = \partial_{x_1} ^2 \left( \log |\, . \,| \right) - \pi \delta_0
\end{equation}
and the Fourier transform of~\eqref{eq:2D dip lem} is given by 
\begin{equation}\label{eq:Fourier high}
 \widehat{U ^{\rm 2D}} (\xi) = \frac{\xi_1 ^2}{|\xi| ^2} - \frac{1}{2}. 
\end{equation}
\end{lemma}

\begin{proof}
The first part is a direct application of~\cite[Theorem~4.12]{Duoandikoetxea-01}, the convolution being well-defined as a principal value because $U^{\rm 2D}$ averages to $0$ on spheres of center $0$. The distributional identity~\eqref{eq:distri high} follows from a straightforward calculation, similar to the standard proof that $x\mapsto - (2\pi) ^{-1} \log |x|$ is the Green function of the Laplacian in 2D, see e.g.~\cite[Theorem~6.20]{LieLos-01}. The 3D analogue of this computation is in~\cite[Appendix~A]{BaoCaiWan-10}. The expression of the Fourier transform follows directly from~\eqref{eq:2D dip high}: since $-i\partial_{x_j}$ is equivalent to multiplication by $\xi_j$ on the Fourier side,  
\begin{equation}\label{eq:formal}
 \widehat{-\log |\, . \,|} (\xi)  = \frac{2\pi}{|\xi| ^2} 
\end{equation}
follows from the distributional identity
$$ - \Delta \left( -\frac{1}{2\pi} \log | \,.\, | \right) = \delta_0.$$
Note that strictly speaking~\eqref{eq:formal} only makes sense when tested against functions with support outside of the origin (in frequency space), but a mild regularization argument allows to deduce~\eqref{eq:Fourier high} nevertheless since the right-hand side is bounded and is tested against bounded functions (Fourier transforms of $L^1$ functions).  
\end{proof}

An immediate consequence of the above and of Plancherel's formula is that we have three equivalent formulations for the high-frequency interaction energy, and continuity thereof. 

\begin{lemma}[\textbf{The high-frequency interaction energy}]\label{lem:ener high}\mbox{}\\
For $\rho \in L^1 \cap L^2$, let 
$$ U_\rho := -\frac{1}{2\pi} \log |\, . \,| $$
be the Coulomb potential generated by $\rho$, solution of 
$$ -\Delta U_\rho = \rho. $$
We have
\begin{align}\label{eq:int high equiv}
F_{a,b} [\rho]&:= a \int_{\R^2} \rho ^2 + b \int_{\R^2} \left( U ^{\rm 2D} \star \rho \right) \rho \nonumber\\
&= \left(a -\frac{b}{2}\right) \int_{\R^2} \rho ^2 - b \int_{\R^2} \left( \partial_{x_1} ^2 U_\rho \right) \rho \nonumber\\ 
&= \left( a - \frac{b}{2} \right) \int_{\R^2} \rho ^2 + b \int_{\R^2} \int_{\R^2} \frac{\xi_1 ^2}{|\xi| ^2}|\hat{\rho} (\xi)|^2 d\xi \\
&= \int_{\R^2}  \frac{(a+b/2) \xi_1 ^2 + (a - b/2) \xi_2 ^2}{|\xi^2|} |\hat{\rho} (\xi)|^2.
\end{align}
Moreover $\rho \mapsto F_{a,b} [\rho]$ is continuous in the strong $L^2$ topology. 
\end{lemma}

\begin{proof}
The continuity follows most easily from the formulation in Fourier variables, the rest is self-evident from the previous lemma. 
\end{proof}

From the last expression in~\eqref{eq:int high equiv}, it is clear that when~\eqref{eq:parameters cond} holds, there exists $u$ with $F_{a,b} [\rho] \geq 0$ (simply concentrate a test function along frequencies close to $\xi_1 = 0$ or $\xi_2 = 0$). Then Inequality~\eqref{eq:GN gen} makes sense and we can proceed to the 

\begin{proof}[Proof of Theorem~\ref{thm:main}, Item (1)]
Clearly, using Plancherel's formula and~\eqref{eq:GN} we have 
$$ \left| F_{a,b} [|u| ^2] \right| \leq C \int_{\R^2} |\nabla u | ^2 \int_{\R^2 }|u|^2$$
so that the infimum 
$$ 
C (a,b) ^{-1} := \inf \left\{ \frac{\int_{\R^2} |\nabla u| ^2 \int_{\R^2} |u| ^2}{ F_{a,b} [|u|^2]} \,\big|\, F_{a,b} [|u|^2] \geq 0 \right\}
$$
exists. We now want to show it is attained. 


Consider a minimizing sequence $(u_n)$ for the above infimum. Denote $\rho_n = |u_n|^2$. By scaling, we may assume that 
$$ \int_{\R^2} |\nabla u_n| ^2  = \int_{\R^2} |u_n| ^2 = 1.$$
It follows from Plancherel and the Sobolev embedding that, for $j=1,2$, 
$$ \int_{\R^2} \frac{\xi_j ^2}{|\xi| ^2} |\hat{\rho_n} (\xi)| ^2 \leq \int_{\R^2} |\hat{\rho_n} (\xi)| ^2 \leq \int_{\R^2} |u_n| ^4 $$
are bounded uniformly. We use the concentration compactness-principle in the form discussed e.g. in~\cite[Appendix]{LenLew-11} and~\cite[Chapter~3]{Lewin-VMQM} (see also~\cite[Section~3.4]{LewRou-13},~\cite[Section~4.2]{KilVis-08} and references therein). We may assume that vanishing (in the sense of the concentration-compactness principle) does not occur, for otherwise $u_n \to 0$ strongly in $L^4$ and thus $F_{a,b}[|u_n|^2] \to 0$. This would lead to the contradiction $C(a,b) = 0$.   

Thus, modulo extraction and translation (we do not relabel, nor indicate the translation, for our variational problem is translation-invariant), we may assume that 
$$u_n = v_n  + r_n + o (1)$$
where $v_n \to u \neq 0$ weakly in $H^1$ and strongly in $L^p, 1<p<\infty$, the distance between the disjoint supports of $v_n$ and $r_n$ goes to infinity when $n\to \infty$, and the $o(1)$ is in $H^1$ norm. Then, clearly,
\begin{align*}
 \int_{\R^2} |\nabla u_n| ^2 &= \int_{\R^2} |\nabla v_n| ^2 + \int_{\R^2} |\nabla r_n| ^2 + o (1)\\
 \int_{\R^2} |u_n| ^2 &= \int_{\R^2} |v_n| ^2 + \int_{\R^2} |r_n| ^2 + o (1)\\
 \int_{\R^2} |u_n| ^4 &= \int_{\R^2} |v_n| ^4 + \int_{\R^2} |r_n| ^4 + o (1)
\end{align*}
and since $U^{\rm 2D}$ decays uniformly at large distances 
$$ \int_{\R^2} \left( U ^{\rm 2D} \star |u_n| ^2 \right) |u_n| ^2 = \int_{\R^2} \left( U ^{\rm 2D} \star |v_n| ^2 \right) |v_n| ^2 + \int_{\R^2} \left( U ^{\rm 2D} \star |r_n| ^2 \right) |r_n| ^2 + o (1),$$
using that the distance between the supports of $v_n$ and $r_n$ diverges.

Then we get easily 
\begin{align}\label{eq:split 1}
 1 &= \int_{\R^2} |\nabla u_n| ^2 \int_{\R^2} | u_n| ^2 \nonumber \\
 &\geq \int_{\R^2} |\nabla v_n| ^2 \int_{\R^2} |v_n| ^2 + \int_{\R^2} |\nabla r_n| ^2 \int_{\R^2} | r_n| ^2 + o (1) \nonumber \\
 &\geq \int_{\R^2} |\nabla v_n| ^2 \int_{\R^2} |v_n| ^2 + C(a,b) ^{-1} F_{a,b} [|r_n| ^2] + o(1). 
\end{align}
On the other hand we obtain  
\begin{multline*}
 F_{a,b} [|v_n| ^2] + C(a,b) \int_{\R^2} |\nabla r_n| ^2 \int_{\R^2} | r_n| ^2  + o (1) \geq F_{a,b} [|v_n| ^2] + F_{a,b} [|r_n| ^2] + o (1) \\ \geq F_{a,b} [|u_n| ^2] = C (a,b) + o(1). 
\end{multline*}
Combining the first inequality with $C(a,b)^{-1}$ times the second we deduce 
$$ C(a,b)^{-1} F_{a,b} [|v_n| ^2] \geq \int_{\R^2} |\nabla v_n| ^2 \int_{\R^2} |v_n| ^2 + o (1).$$ 
Passing to the liminf in the right-hand side and using the strong $L^p$ convergence in the left-hand side we conclude that 
$$ C(a,b)^{-1} F_{a,b} [|v| ^2] \geq \int_{\R^2} |\nabla v| ^2 \int_{\R^2} |v| ^2 $$ 
and thus $v$ is the sought-after optimizer.
\end{proof}

\begin{remark}[Strong convergence of minimizing sequences]\label{rem:strong CV}\mbox{}\\
In~\eqref{eq:split 1} we have dropped the positive cross-terms 
$$ \int_{\R^2} |\nabla v_n| ^2 \int_{\R^2} |r_n| ^2 + \int_{\R^2} |\nabla r_n| ^2 \int_{\R^2} | v_n| ^2$$
from the lower bound. Retaining them we have the additional information that they must converge to $0$ in the limit, hence that $v_n \to v$ strongly in $H^1$. This convergence holds (modulo translation and extraction) for any minimizing sequence having 
$\int_{\R^2} |\nabla v_n| ^2 = \int_{\R^2} |v_n| ^2 = 1.$ 
\hfill\qed
\end{remark}

\section{Proof of the main stability criterion}\label{sec:proof main}

Let us start by relating the true, physical, interaction~\eqref{eq:2D int II} to its high frequency limit. Since the stability/instability criterion is related to whether or not minimizing sequences collapse to a point we need the following:

\begin{lemma}[\textbf{High frequency limit}]\label{lem:high lim}\mbox{}\\
Let $(\nu_n)$ be a bounded sequence in $L^1 \cap L^2$, $(L_n)$ a sequence of positive numbers with $L_n \to 0$ and  
$$\rho_n (x) := L_n ^{-2} \nu_n \left(\frac{x}{L_n}\right).$$  
We have, for the physical 2D interaction~\eqref{eq:F int}, 
\begin{equation}\label{eq:int collapse}
F^{\rm int} [ \rho_n] = \frac{2}{L_n^2} \int_{\R^2} \frac{n_3 ^2 \xi_2 ^2 - (1-2n_3^2)\xi_1 ^2}{|\xi|^2}|\hat{\nu_n}(\xi)| ^2 d \xi + O(L_n ^{-1})
\end{equation}
in the limit $n\to \infty$.
\end{lemma}

\begin{proof}
We use the expression in Fourier space~\eqref{eq:2D int freq}, which gives, after scaling, 
\begin{equation}\label{eq:int scale} 
F^{\rm int} [\rho_n] = \frac{1}{\pi L_n ^2}\int_{\R} \int_{\R^2} \frac{n_3 ^2 \xi_2 ^2 - (1-2n_3^2)\xi_1 ^2}{|\xi|^2 + s^2 L_n ^2}|\hat{\nu_n}(\xi)| ^2 e^{-s^2 / 4\pi }d\xi ds.
\end{equation}
We split the $\xi$ integral according to whether $|\xi| \leq \eps$ or the other way around. The former contribution is $O(\eps ^2 L_n ^{-2})$ since $(\nu_n)$ is bounded in $L^1$ and thus $(\hat{\nu_n})$ is bounded in $L^\infty$. For the latter contribution we use that 
$$ \left| \frac{1}{|\xi| ^2 + s^2 L_n ^2} - \frac{1}{|\xi| ^2}\right| \leq \frac{s^2 L^2}{|\xi| ^ 4}.$$
This gives 
\begin{multline*}
\int_{\R} \int_{B(0,\eps) ^c} \frac{n_3 ^2 \xi_2 ^2 - (1-2n_3^2)\xi_1 ^2}{|\xi|^2 + s^2 L_n ^2}|\hat{\nu_n}(\xi)| ^2 e^{-s^2 / 4\pi }d\xi ds \\= 2 \pi \int_{B(0,\eps) ^c} \frac{n_3 ^2 \xi_2 ^2 - (1-2n_3^2)\xi_1 ^2}{|\xi|^2 }|\hat{\nu_n}(\xi)| ^2 d\xi + O \left( L_n ^2 \int_{B(0,\eps) ^c} \frac{1}{|\xi|^2}|\hat{\nu_n}(\xi)| ^2 d\xi \right)\\
= 2 \pi \int_{\R^2} \frac{n_3 ^2 \xi_2 ^2 - (1-2n_3^2)\xi_1 ^2}{|\xi|^2 }|\hat{\nu_n}(\xi)| ^2 d\xi + O (\eps ^2) + O (L_n^{2} \eps ^{-2})
\end{multline*}
where we bound $|\xi| ^{-2}$ by its sup over $B(0,\eps) ^c$, use that $(\nu_n)$, hence $(\hat{\nu_n})$ is bounded in $L^2$, and get rid of the contribution of $B(0,\eps)$ as precedently. There remains to optimize the estimate by picking $\eps = L_n^{1/2}$.
\end{proof}

Now we can complete the 

\begin{proof}[Proof of Theorem~\ref{thm:main}]
We start with the existence result. Let $(u_n)$ be a minimizing sequence. If it is bounded in $H^1 (\R^2)$, we extract a weakly convergent subsequence in $H^1$ and may conclude immediately: using~\cite[Lemma~2.2]{BaoAbdCai-12} (analogue of Lemma~\ref{lem:ener high}, but for the original interaction~\eqref{eq:F int}), the interactions terms are bounded below uniformly. We deduce that the potential term $\int_{\R^2} V |u_n| ^2$ is also uniformly bounded. Since $V$ grows at infinity, we may then use a compact embedding yielding strong convergence in $L^p$, $2< p < \infty$ along a subsequence. The interaction terms then converge by~\cite[Lemma~2.2]{BaoAbdCai-12} again, and we may pass to the liminf in the kinetic energy using Fatou's lemma. All in all, for the weak limit $v$ we get 
$$ E_{\rm 2D} \geq \cE_{\rm 2D} [v]$$
and by the strong $L^2$ convergence, $v$ is $L^2$-normalized. Thus we obtain an optimizer for Problem~\eqref{eq:2D min ener}. 

Let us then work under the assumptions of Item (2) and prove by contradiction that any minimizing sequence must be bounded in $H^1 (\R ^2)$. Let $(u_n)$ be such a minimizing sequence and 
$$L_n = \left(\int_{\R^2} |\nabla u_n| ^2\right) ^{-1/2}.$$  
Assume that there exists a subsequence (not relabeled) along which $L_n \to 0$. Then define
$$ v_n (x):= L_n  u_n \left( x L_n \right).$$
By definition we have 
$$ \int_{\R^2} |\nabla v_n| ^2 = \int_{\R^2} | v_n| ^2 = 1$$
and 
$$
\int_{\R^2} | u_n| ^4 = L_n ^{-2} \int_{\R^2} |v_n| ^4.
$$
Moreover $\rho_n = |u_n| ^2$ satisfies the assumptions of Lemma~\ref{lem:high lim} with $\nu_n = |v_n| ^2$. Therefore 
$$ \cE_{\rm 2D} [u_n] \geq L_n ^{-2} \left( \int_{\R ^2} |\nabla v_n| ^2  - F_{a,b} [|v_n| ^2] \right) + O(L_n ^{-1})$$
with the parameters $a,b$ chosen as in~\eqref{eq:parameters}. In case~\eqref{eq:param triv}, it follows from Lemma~\ref{lem:ener high} that $F_{a,b} [|v_n| ^2] \leq 0$. In the other case we get that   
$$ 
F_{a,b} [|v_n| ^2] \leq C(a,b) \int_{\R^2} |\nabla v_n| ^2 \int_{\R^2} | v_n| ^2 < \int_{\R^2} |\nabla v_n| ^2  \int_{\R^2} | v_n| ^2. 
$$
In both cases there is a constant $c>0$ such that
$$\cE_{\rm 2D} [u_n] \geq c L_n ^{-2} \int_{\R ^2} |\nabla v_n| ^2  + O(L_n ^{-1})$$
which is a contradiction if $L_n \to 0$ because $\cE_{\rm 2D} [u_n]$ is bounded above and $\int_{\R ^2} |\nabla v_n| ^2 = 1.$ We deduce that the sequence $(u_n)$ was in fact bounded in $H^1$ in the first place, and we may conclude as discussed at the beginning of the proof.

\medskip

For the non-existence result, it suffices to construct a trial state whose energy tends to $-\infty$. Let again the parameters $a,b$ be chosen according to~\eqref{eq:parameters} and $u_{a,b}$ be the optimizer for Inequality~\eqref{eq:GN gen}, proven to exist in the preceding section. By scaling invariance we may choose it to be $L^2$-normalized. Define 
$$ u_L := L ^{-1} u_{a,b} \left( \frac{x}{L}\right)$$
for some sequence of lengths $L \to 0$. Clearly 
$$ \int_{\R^2 } V |u_L| ^2 \to V (0).$$
For the rest of the energy we get, using Lemma~\ref{lem:high lim} again, 
\begin{align*}
 \cE_{\rm 2D} [u_L] &= L^{-2} \left( \int_{\R^2} |\nabla u_{a,b}| - F_{a,b} [|u_{a,b}| ^2] \right) + O (L^{-1}) + V(0) + o (1)\\
 &= L^{-2} \left(1 - C(a,b) \right) \int_{\R^2} |\nabla u_{a,b}| ^2 + O (L^{-1})
\end{align*}
and this tends to $-\infty$ when $L\to 0$ under the conditions of Item (3).
\end{proof}

\section{Borderline cases}\label{sec:borderline}

We start by refining Lemma~\ref{lem:high lim}, taking into account subleading corrections: 

\begin{lemma}[\textbf{High frequency limit, refinement}]\label{lem:high lim ref}\mbox{}\\
With the notation of Lemma~\ref{lem:high lim ref}, we have, for the physical 2D interaction~\eqref{eq:F int}, 
\begin{multline}\label{eq:int collapse ref}
F^{\rm int} [ \rho_n] = \frac{2}{L_n^2} \int_{\R^2} \frac{n_3 ^2 \xi_2 ^2 - (1-2n_3^2)\xi_1 ^2}{|\xi|^2}|\hat{\nu_n}(\xi)| ^2 d \xi \\
+ 4\pi \int_{B(0,L_n) ^c} \frac{(1-2n_3^2)\xi_1 ^2 - n_3 ^2 \xi_2 ^2}{|\xi|^4}|\hat{\nu_n}(\xi)| ^2 d \xi + O(1)
\end{multline}
in the limit $n\to \infty$ (i.e. $L_n \to 0$).
\end{lemma}

\begin{proof}
Starting from~\eqref{eq:int scale}, we again split the $\xi$ integral according to whether $|\xi| \leq \eps$ or the other way around. Again, the small frequency part contributes a  $O(\eps ^2 L_n ^{-2})$. For the high frequency part we use that  
$$ \left| \frac{1}{|\xi| ^2 + s^2 L_n ^2} - \frac{1}{|\xi| ^2} + \frac{s^2 L_n^2}{|\xi| ^4}\right| \leq \frac{s^4 L_n^4}{|\xi| ^ 6}.$$
This gives 
\begin{align*}
\int_{\R} \int_{B(0,\eps) ^c} \frac{n_3 ^2 \xi_2 ^2 - (1-2n_3^2)\xi_1 ^2}{|\xi|^2 + s^2 L_n ^2} &|\hat{\nu_n}(\xi)| ^2 e^{-s^2 / 4\pi }d\xi ds \\
&= 2 \pi \int_{\R^2} \frac{n_3 ^2 \xi_2 ^2 - (1-2n_3^2)\xi_1 ^2}{|\xi|^2 }|\hat{\nu_n}(\xi)| ^2 d\xi + O (\eps ^2 )\\
& - 4 \pi^2 L_n ^2 \int_{B(0,\eps) ^c} \frac{n_3 ^2 \xi_2 ^2 - (1-2n_3^2)\xi_1 ^2}{|\xi|^4 }|\hat{\nu_n}(\xi)| ^2 d\xi \\& + O \left( L_n ^4 \int_{B(0,\eps) ^c} \frac{1}{|\xi| ^4} |\hat{\nu_n} (\xi)| ^2 d\xi \right).
\end{align*}
Using that $(\hat{\nu_n})$ is uniformly bounded in $L^{\infty}$ (because $(\nu_n)$ is uniformly bounded in $L^1$) we conclude that the last term is a $O(L_n ^4 \eps ^{-2})$. Choosing $\eps = L_n$ optimizes the estimate and yields the lemma. 
\end{proof}

\begin{proof}[Proof of Theorem~\ref{thm:border}]
We prove the statements in the order they appear in the Theorem, but the main argument is anyway the same in all cases.

\medskip

\noindent \textbf{Proof of Item 1.} As in the proof of Item (2) in Theorem~\ref{thm:main}, we mostly need prove that minimizing sequences must be bounded in $H^1$, and the existence result follows easily. Let then $(u_n)$ be a minimizing sequence and 
$$L_n = \left(\int_{\R^2} |\nabla u_n| ^2\right) ^{-1/2}.$$  
Assume that there exists a subsequence (not relabeled) along which $L_n \to 0$. Then define
\begin{equation}\label{eq:scale sequence}
 v_n (x):= L_n  u_n \left( x L_n \right).
\end{equation}
By definition we have 
$$ \int_{\R^2} |\nabla v_n| ^2 = \int_{\R^2} | v_n| ^2 = 1.$$
We use Lemma~\ref{lem:high lim ref} with $\rho_n = |u_n| ^2$ and $\nu_n = |v_n| ^2$: 
\begin{align}\label{eq:border stab}
 \cE_{\rm 2D} [u_n] &\geq L_n ^{-2} \left( \int_{\R ^2} |\nabla v_n| ^2  - F_{a,b} [|v_n| ^2] \right) \nonumber
 \\& - 3\pi \lambda \int_{B(0,L_n) ^c} \frac{(1-2n_3^2)\xi_1 ^2 - n_3 ^2 \xi_2 ^2}{|\xi|^4}|\hat{\nu_n}(\xi)| ^2 d \xi + O(1) \nonumber
 \\&\geq - 3\pi \lambda \int_{B(0,L_n) ^c} \frac{(1-2n_3^2)\xi_1 ^2 - n_3 ^2 \xi_2 ^2}{|\xi|^4}|\hat{\nu_n}(\xi)| ^2 d \xi + O(1)  
\end{align}
with the parameters $a,b$ chosen as in~\eqref{eq:parameters} and using~\eqref{eq:GN gen} with $C(a,b) = 1$. The term in the last line is $O(|\log L_n|)$ because $(\hat{\nu_n})$ is uniformly bounded in $L^{\infty}$. Since the energy is bounded from above and $L_n \to 0$ by assumption, multiplying the previous inequality by $L_n ^2$ gives that  
$$ \int_{\R ^2} |\nabla v_n| ^2  - F_{a,b} [|v_n| ^2] \underset{n\to \infty}{\to} 0.$$
Hence $v_n$ is an optimizing sequence for~\eqref{eq:GN gen}. As per Remark~\ref{rem:strong CV}, we may after extraction and translation assume it converges strongly in $H^1$. We do not relabel the extraction and ignore the translation, for it affects only the potential term of the energy that we have already dropped from the lower bound.

For shortness, we denote 
$$ f(\xi) = \frac{(1-2n_3^2)\xi_1 ^2 - n_3 ^2 \xi_2 ^2}{|\xi|^4}.$$
Pick now $ L_n \ll \ell_n = L_n ^{\alpha} \ll 1 $ with $\alpha$ a small positive number to be fixed later on. Clearly, 
\begin{equation}\label{eq:a_n}
 a_n := \int_{L_n \leq |\xi| \leq \ell_n} |\xi| \left|  f(\xi) \right| d\xi \underset{n\to \infty}{\to} 0. 
\end{equation}
On the other hand, since $(\nu_n)$ converges strongly in $L^1$, it is a tight sequence, and modulo extracting a subsequence, we may assume that (see e.g.~\cite[Lemma~3.8]{Lewin-VMQM}) 
\begin{equation}\label{eq:tightness} 
\int_{|x| \geq a_n^{-1/2}} \nu_n (x) dx \underset{n\to \infty}{\to} 0.
\end{equation}
Then, we split the last term in~\eqref{eq:border stab} as follows:
\begin{align}\label{eq:border split}
\int_{B(0,L_n) ^c} &\frac{(1-2n_3^2)\xi_1 ^2 - n_3 ^2 \xi_2 ^2}{|\xi|^4}|\hat{\nu_n}(\xi)| ^2 d \xi\nonumber
\\&= \int_{ 1 \leq |\xi|} f(\xi) |\hat{\nu_n} (\xi)| ^2 dx dy d \xi\nonumber
\\&+ \int_{ \ell_n \leq  |\xi| \leq 1} \iint_{\R^2 \times \R ^2}\nu_n (x) \nu_n (y) e^{i\xi \cdot(x-y)} f(\xi)  dx dy d \xi\nonumber
\\&+ \int_{ L_n \leq |\xi| \leq \ell_n} \iint_{|x-y| \leq a_n^{-1/2}} \nu_n (x) \nu_n (y) e^{i\xi \cdot(x-y)}  f(\xi)   dx dy d \xi\nonumber
\\&+ \int_{ L_n \leq |\xi| \leq \ell_n} \iint_{|x-y| \geq a_n ^{1/2}} \nu_n (x) \nu_n (y) e^{i\xi \cdot(x-y)} f(\xi)  dx dy d \xi\nonumber
\\&= \mathrm{I} + \mathrm{II} + \mathrm{III} + \mathrm{IV}.
\end{align}
Term number $\mathrm{I}$ we just bound by a fixed constant, using that $\hat{\nu_n}$ is uniformly bounded in $L^2$. Since $\int \nu_n = 1$ it is immediate to see that 
$$ |\mathrm{II}| \leq \int_{\ell_n \leq |\xi| \leq 1} |f(\xi)| d\xi \leq C | \log \ell_n | = C \alpha |\log L_n|.$$
By a similar bound, but using now~\eqref{eq:tightness}, we get 
$$ |\mathrm{IV}| \leq C |\log L_n| \iint_{|x-y| \geq a_n ^{1/2}} \nu_n (x) \nu_n (y) dx dy \ll |\log L_n|.$$
As for $\mathrm{III}$, by a first order Taylor expansion of $e^{i\xi \cdot(x-y)}$, 
\begin{align*}
 \int_{ L_n \leq |\xi| \leq \ell_n} \iint_{|x-y| \leq a_n^{-1/2}} &\nu_n (x) \nu_n (y) e^{i\xi \cdot(x-y)}  f(\xi)  dx dy d \xi  \\
 &= \int_{ L_n \leq |\xi| \leq \ell_n} \iint_{|x-y| \leq a_n^{-1/2}} \nu_n (x) \nu_n (y)   f(\xi)   dx dy d \xi\\
 &+ O\left( \int_{ L_n \leq |\xi| \leq \ell_n} \iint_{|x-y| \leq a_n^{-1/2}} \nu_n (x) \nu_n (y) |x-y| |\xi|  |f(\xi)|  dx dy d \xi \right).
\end{align*}
Recalling~\eqref{eq:a_n} the last term is bounded by $a_n ^{1/2} \to 0$ when $n\to \infty$. For the first one we perform explicitly the $\xi$ integration:
\begin{align*}
 \int_{ L_n \leq |\xi| \leq \ell_n} f(\xi) d\xi &= \int_{ r = L_n} ^{\ell_n} \int_{\theta = 0} ^{2\pi} \frac{(1-2n_3 ^2) \cos^2 \theta - n_3^2 \sin ^2 \theta}{r^2} rdr d\theta \\
 &= (1-3 n_3 ^2 ) (\log \ell_n - \log L_n) \int_{0} ^{2\pi} \cos ^2 (\theta) d\theta  \\
 &= \pi (1-3 n_3 ^2 ) (\alpha - 1)  \log L_n.
\end{align*}
Combining with~\eqref{eq:tightness} we get 
$$ \mathrm{III} \sim \pi (1-3 n_3 ^2 ) (\alpha - 1) \log L_n.$$
All in all, since $\alpha$ can be chosen arbitrarily small, we may return to~\eqref{eq:border stab} and conclude that 
$$ \cE_{\rm 2D}[u_n] \geq c \lambda (1-3n_3 ^2) \log L_n  + O (1)$$
with $c>0$ a constant. Then, if $\lambda (1-3n_3 ^2) < 0$ this certainly implies that the energy must be bounded below and rules out the possibility that $L_n \to 0$. Minimizing sequences must then be bounded in $H^1$ and we may conclude the proof in a standard way.

\medskip

\noindent \textbf{Proof of Item 2.} Assume now that $\lambda (1-3n_3 ^2) = 0$. Then either $\lambda = 0$ and the result is already included in~\cite{GuoSei-13} or $1-3n_3 ^2 = 0$ and we return to~\eqref{eq:border stab}. In this case
$$ f(\xi) = \frac{(1-2n_3^2)\xi_1 ^2 - n_3 ^2 \xi_2 ^2}{|\xi|^4}$$
averages to $0$ on spheres centered at the origin, so that term number $\mathrm{III}$ in~\eqref{eq:border split} is uniformly bounded, its limit $L_n \to 0$ being well-defined as a principal value. Indeed, following~\cite[Chapter~4]{Duoandikoetxea-01}, we have that the inverse Fourier transform of $f(\xi)$ is 
$$ c \delta_0 + c' \frac{x_1 ^2}{|x| ^2}$$
for two constants $c$ and $c'$ (this is similar to what we did in Lemma~\ref{lem:int high}). Using Plancherel's formula, the last term in~\eqref{eq:border stab} converges when $L_n \to 0$ to a term of the form  
$$ c \int_{\R^2} |v_n| ^4 + c' \int_{\R^2\times \R ^2} |v_n (x)| ^2 \frac{(x_1-y_1) ^2 }{|x-y|^2} |v_n (y)| ^2 dx dy,$$
which is clearly bounded for $v_n$ bounded in $H^1$. We conclude that the energy is bounded below. However, taking a trial sequence $(u_n)$ of the form~\eqref{eq:scale sequence} with $L_n \to$ and  $v_n \equiv u_{a,b}$ independent of $n$, the above discussion shows that its energy stays is bounded despite it collapsing to a point.   

\medskip

\noindent \textbf{Proof of Item 3.} Take again a trial sequence $(u_n)$ of the form~\eqref{eq:scale sequence} with $L_n \to 0$ and  $v_n \equiv u_{a,b}$. The energy $\cE_{\rm 2D} [u_n]$ is then given by the right-hand side of~\eqref{eq:border stab} again (the potential term is of lower order). The last term can be analyzed exactly as previously, and we obtain that it is bounded above by 
$$ c \lambda (1-3n_3 ^2) \log L_n  + O (1).$$
With $\lambda (1-3n_3 ^2) >0$ this goes to $-\infty$ as $n\to \infty$, and this concludes the proof.
\end{proof}

%

\end{document}